\documentclass[10pt, twocolumn, journal]{IEEEtran}

\usepackage{amsmath,amssymb,amsfonts}
\usepackage{graphicx}
\usepackage{algorithm}
\usepackage{cite}
\usepackage{algpseudocode}
\usepackage{bm}
\usepackage{color}
\usepackage{amsthm}
\usepackage{subcaption}

\allowdisplaybreaks

\newtheorem{Prop}{Proposition}
\newtheorem{Lem}{Lemma}

\newtheorem{Def}{Definition}

\DeclareMathOperator{\Tr}{Tr}

\begin{document}

\title{Task Offloading and Resource Allocation with Multiple 
CAPs and Selfish Users}

\author{\IEEEauthorblockN{Eric Jiang, 
Meng-Hsi Chen, Ben Liang, and  Min Dong}
\thanks{Eric Jiang and Ben Liang are with the Department of 
Electrical and Computer Engineering, University of Toronto, Toronto, Canada 
(e-mail: eric.jiang@mail.utoronto.ca, liang@ece.utoronto.ca)}
\thanks{Meng-Hsi Chen is with MediaTek, Hsinchu, Taiwan (email: 
menghsichen@gmail.com)}
\thanks{Min Dong is with the Department of Electrical, Computer and Software    
Engineering,  Ontario  Tech University, Oshawa, Canada 
(email: min.dong@ontariotechu.ca)}
\thanks{Part of this manuscript has been presented in 
        \cite{chen2015cloudnet}.}
}

\maketitle

\begin{abstract}
In this work, we consider a multi-user mobile edge computing system with 
multiple computing access points (CAPs).  Each mobile user has multiple 
dependent tasks that must be processed in a round-by-round schedule.  In every 
round, a user may process their individual task locally, or choose to offload 
their task to one of the $M$ CAPs or the remote cloud server, in order to 
possibly reduce their processing cost. 
We aim to jointly optimize the offloading decisions of the users and the  
resource allocation decisions for each CAP over a global 
objective function, defined as a weighted sum of total energy consumption and 
the round time. We first present a centralized heuristic solution, termed MCAP, 
where the original problem is relaxed to a semi-definite program (SDP) to 
probabilistically generate the offloading decision. Then, recognizing that the 
users often exhibit selfish behavior to reduce their individual 
cost, we propose a game-theoretical approach, termed MCAP-NE, which allows us 
to compute a Nash Equilibrium (NE) through a finite improvement method starting 
from the previous SDP solution. This approach leads to a solution from which 
the users have no incentive to deviate, with substantially reduced NE 
computation time.
In simulation, we compare the system cost of the 
NE solution with those of MCAP, MCAP-NE, a random mapping, and the optimal 
solution, showing that our NE solution attains near optimal performance under a 
wide 
set of parameter settings, as well as demonstrating the advantages of using 
MCAP to produce the initial point for MCAP-NE.

\end{abstract}

\section{Introduction}

With the ever-increasing demand for computational resources by mobile 
applications, from a variety of tasks including machine learning, virtual 
reality, and natural language processing, there is a growing need for 
accessible computational resources in mobile networks that are available for 
task offloading \cite{kumar2013}.  Mobile Cloud Computing (MCC) would allow 
mobile devices to offload tasks to large remote data centres in the cloud, 
such as Amazon EC2 \cite{amazonec2}.  This promises to improve the performance 
of demanding mobile applications and expand the computational capabilities of 
mobile devices for future applications. 

However, MCC systems may cause high latency for these mobile users, due to the 
transmission delays required to wirelessly offloading tasks through long 
distances in order to reach the remote server.  This can prove fatal in 
providing adequate quality of service of many applications, especially 
computationally-intensive real-time ones such as virtual 
reality. One means of resolving this is by providing 
additional wireless and/or computational resources at the edge of the mobile 
system, such as the wireless base station, thus giving the users faster access 
to the necessary resources \cite{MECSurvey2}, \cite{5GBook}, 
\cite{etsi2016framework}. Because base stations have high-speed 
connections to the Internet and thus the various cloud servers, offloading 
through these edge hosts could substantially reduce the transmission time and 
energy required to access these servers.  Additionally, base stations may 
hold their own servers, providing mobile users with an additional site for 
possible task offloading. When 
computational resources are installed at or near radio access networks, the 
system paradigm is known as Mobile 
Edge Computing (MEC), as defined by the European Telecommunications Standards 
Institute (ETSI) \cite{etsi2016framework}.  However, other similar systems, 
such as micro cloud centers, \cite{greenberg2008}, cloudlets 
\cite{cloudlets2009}, and fog computing \cite{fog2012}, use similar methods to 
reduce the communication latency that result from MCC.   

When these computing resources are built into a wireless access point or 
cellular base station, we refer to it as a \emph{computing access point} 
(CAP) \cite{chen2016icassp}, \cite{chen2018}, \cite{chen2017}, 
\cite{chen2018wireless}.  The CAP performs two 
functions---it 
serves as a link between the individual mobile devices and the 
remote cloud server, whereby the application can be forwarded through a 
high-speed 
connection, and it serves as an additional site for the computation of the task 
itself, if doing so is more beneficial than local or remote processing.  

One research problem that arises from MCC and MEC systems is that of the 
offloading 
decision---how to determine whether a task should be executed on the user's 
mobile device, or offloaded to the available edge or cloud servers 
\cite{5GBook}.  Additionally, the available bandwidth and 
computational resources are limited, especially at the CAP---thus good system 
performance would require a judicious distribution of these resources between 
the users.  Ideally, the offloading decision and the resource allocation would 
be jointly optimized in order to maximize the overall quality of service 
provided by the system to the users.  However, this is a difficult problem due 
to the multiple tiers of offloading available, the heterogeneous resources that 
must be assigned, the asymmetrical nature of the system, and the need to 
optimize over a set of binary offloading decisions, resulting in a 
mixed-integer programming problem that is non-convex and difficult to solve.

In \cite{chen2016icassp}, a single CAP and remote cloud 
server system were considered, with multiple users and one atomic task per user 
to be offloaded. A centralized controller was designed to jointly optimize the 
offloading decisions and resource allocations to minimize a weighted sum of 
total energy consumption and the time required to process all of the 
tasks (i.e. the round time). The problem was formulated as a mixed-integer 
optimization problem, and using semidefinite relaxation (SDR) techniques from 
\cite{luo2010}, a heuristic solution was presented and shown 
through simulation to perform near optimally under a wide range of system 
parameters. 

However, in practical systems often multiple CAPs are simultaneously available 
to the users. This added dimension in the solution space adds significant 
complexity to the offloading decision and resource allocation decisions, 
especially since the computational capability of different CAPs, as well as the 
quality of wireless access to them, can substantially differ. Furthermore, the 
centralized optimization model does not account for any 
agency of the users in making these offloading decisions, who often exhibit 
selfish behaviour in reality, where each individual user  
would choose the offloading site that minimizes their own individual cost.  In 
such cases, individual users would not have an incentive to follow the 
decisions of the central controller, limiting the scope of 
\cite{chen2016icassp}. 

In this work, we consider a multi-user system with a remote cloud server and 
multiple CAPs.  
Each user has a single task per round, which may be processed 
at the user's mobile device, or offloaded it to one of the CAPs, where it may 
be processed either directly or further offloaded to the remote cloud server.   
The goal is to choose a 
set of offloading and resource allocation decisions that minimize the system 
cost or objective, which we define as a weighted sum of energy consumption and 
the round time, as in \cite{chen2016icassp} and \cite{chen2018}.  We further 
aim 
to ensure that selfish users have no incentive to deviate from their prescribed 
offloading decisions.

The contributions of our work are as follows:
\begin{itemize}
        
        \item We develop the MCAP heuristic by modelling the above 
        optimization problem as a 
        quadratically-constrained quadratic program (QCQP), extending the 
        single-CAP formulation in \cite{chen2016icassp} to multiple CAPs.  To 
        do so, we must make significant changes to the formulation in 
        \cite{chen2016icassp}, particularly through new objective variables 
        and 
        constraints to accommodate the additional offloading possibilities.  We 
        also accommodate for the addition of placement constraints, which prevent 
        user from offloading their task through some CAPs.  We then 
        relax the problem to produce a semidefinite program (SDP), and use those 
        results 
        to probabilistically produce a series of offloading decisions, following 
        the work in \cite{luo2010}.  We choose the decisions among those trials 
        that minimizes the system cost.
        
        \item We then consider the phenomenon of selfish users, where 
        individual 
        users in the system may deviate from the centralized decisions in order to 
        minimize their individual costs.  Utilizing game 
        theory, we show that the multi-CAP system with selfish users can be 
        modeled as a strategic finite game with an ordinal potential function, 
        which allows us to compute a Nash Equilibrium (NE) through the finite 
        improvement method \cite{monderer1996}.  While a 
        game-theoretic approach may suffer from an additional price of anarchy, the 
        NE provides overall system performance that is close to the optimal 
        solution.
        
        \item However, 
        due to the 
        combinatorial nature of the finite improvement method, solving for the NE 
        may 
        require substantial computational time.  To reduce the number of 
        iterations required, we propose using a starting point that is 
        closer to the optimal solution than a randomly chosen one.  We use MCAP to 
        obtain such a starting point, which is then combined with the finite 
        improvement method to construct the MCAP-NE solution. Thus, MCAP-NE 
        improves on 
        MCAP by further reducing the system cost and
        accounting for the agency of selfish users.  Simulation 
        results demonstrate the 
        superiority of the SDR approach in MCAP over a random set of offloading 
        decisions, 
        the     further improvements from the game-theoretic approach in MCAP-NE, the 
        close 
        proximity of these solutions to the optimal solution of the system, and the 
        computational improvements due to the MCAP starting point, under a 
        wide range of parameter settings.
        
\end{itemize}

Our paper is organized as follows.  Section 2 reviews related works in 
solving the offloading problem in MEC systems.
Section 3 formulates the problem and system model, culminating in a 
centralized optimization problem.  Section 4  
reformulates the system as a QCQP, which is then optimized heuristically 
through an SDR approach, to be used as the initial starting point for our NE
solution method.
Section 5 then reformulates the problem as 
a strategic form game, proves the existence of an ordinal 
potential function and NE, and presents an algorithm to find the NE.  Section 
6 presents our numerical results.  Finally, Section 7 concludes the paper.

\section{Related Works}
While there are many existing works that study the offloading problem in 
two-tier cloud systems, including \cite{zhang2013}, \cite{barbarossa2013}, 
\cite{zhang2013infocom}, \cite{chen2016icc} and \cite{munoz2015}, among 
others, fewer works have studied three tier offloading networks.  Such works 
include \cite{chen2016icassp}, 
\cite{chen2018}, \cite{chen2017}, \cite{chen2018wireless}, \cite{Rahimi2012}, 
\cite{Rahimi2013}, \cite{Song2014}, and \cite{MathProg}.  However, 
\cite{Rahimi2012}, \cite{Rahimi2013}, \cite{Song2014}, and \cite{MathProg} only 
consider the offloading decision and do not attempt to optimize the resource 
allocation.
Joint optimization over the offloading decisions and the resource allocation 
is studied in \cite{chen2016icassp} (and its extension in \cite{chen2018}) and 
\cite{chen2017} (and its extension in 
\cite{chen2018wireless}) where a multi-user, single 
CAP system with a remote cloud server is considered.  These works present a 
heuristic centralized solution using SDR for the offloading of one and multiple 
tasks per round respectively, with delay constraints also being considered in 
\cite{chen2018}.  As stated above, in this work we consider the more complex 
problem of multiple CAPs and add the assumption of
selfish users to solve the problem using both SDR and game theory.

Game-theoretic approaches to analyzing mobile offloading networks have been 
presented in \cite{MathProg}, \cite{wang2013sose}, \cite{Swede}, 
\cite{Meskar2015}, \cite{Ma}, \cite{Chen2014}, \cite{chen2015efficient}, 
\cite{Swede2019}, and \cite{ML}.  All of these works except \cite{MathProg} and 
\cite{wang2013sose} 
consider a single 
atomic task per user, which is similar to our work, while \cite{MathProg} and 
\cite{wang2013sose} stand 
apart for considering a Poisson generation of user tasks and a queueing system 
for offloading.
All of these works except \cite{Swede} study systems with potential 
functions and utilize that fact in the computation of an NE. While the 
system in
\cite{Swede} is 
not a potential game, their proof of the existence of an NE relies on the fact 
that a subgame within their system is a potential game, and they utilize that 
fact to modify the finite improvement method in order to produce an algorithm 
that could find the NE of the overall game. 

In \cite{Meskar2015} and \cite{Ma}, 
the finite improvement method from \cite{monderer1996} is directly 
adopted.   These works however do not 
consider the practical implementation of the finite improvement method, which 
is expressly considered in \cite{Chen2014} and \cite{chen2015efficient}, 
where a decentralized solution was presented.  In these works, each user 
computes their improvements locally, using a pilot signal to determine the 
the interference at the wireless channel (which entails all of the necessary 
information to the user), and a 
base station that coordinates the transmission of this signal between the 
users.  A similar method is proposed in \cite{Swede2019}, where the 
computation of the improvements are computed by the individual users, 
with a central controller coordinating these computations between the users and 
disclosing the relevant information to them.  The finite improvement 
method is further improved in \cite{Swede2019} by demonstrating theoretical 
limits to the available strategies for each user that may result in an 
improvement.  Another decentralized approach  
is considered in \cite{ML}, where each user first adopts a mixed strategy, 
then uses reinforcement learning to 
converge to a pure strategy NE.  Thus, \cite{ML} does not use the finite 
improvement method utilized by the other above works to compute the 
NE---however, \cite{ML} still relies on the 
existence of an ordinal potential function to demonstrate the convergence of 
their solution method to the NE.  

In this work, we adopt a centralized approach 
to computing the NE, 
which allows the system to compute an initial starting point to reduce the 
overall computational time.  Furthermore, we demonstrate that our system is 
strategy-proof, thus ensuring that users have no incentive to provide false 
information to the controller, which guarantees the viability of a centralized 
approach.  More importantly, none of these works considers a three-tier 
computing systems. In particular, in the offloading systems of 
\cite{Meskar2015} and \cite{Chen2014} (with a single wireless access point), or 
\cite{Swede}, 
\cite{Ma}, \cite{chen2015efficient}, 
\cite{Swede2019}, and \cite{ML} (with multiple wireless access points), the 
wireless access points serve only to forward the offloaded tasks to the cloud. 
Therefore, their offloading and resource allocation solutions are not 
applicable to our problem.

\section{System Model and Centralized Optimization}

In this section, we develop the system model in question, denoting all the 
relevant variables and modeling the cost of processing at every 
offloading site.  From this, we arrive at a mixed-integer programming problem 
to minimize the overall cost of the system.

Consider a cloud access network consisting of one remote cloud
server, $M$ CAPs denoted by the set $\mathcal{M} \in \{1,\ldots, M\}$, and $N$ 
mobile users denoted by the set $\mathcal{N} \in \{1, \ldots, N\}$, as shown in 
Fig. \ref{fig:system_model}.  
Each mobile user may have multiple tasks to be processed, and we consider a
round-by-round schedule where one task from each user is processed in each 
round.  No task in the next round may begin processing 
until all the tasks of the current round have been processed in the system.
Because of this condition, it suffices to optimize the offloading decisions and 
resources allocation for mobile users in a single round. This will be the focus 
of the remainder of the work, and without loss of generality we assume there 
are $N$ tasks in this round.

\begin{figure}
        \includegraphics[width=\columnwidth]{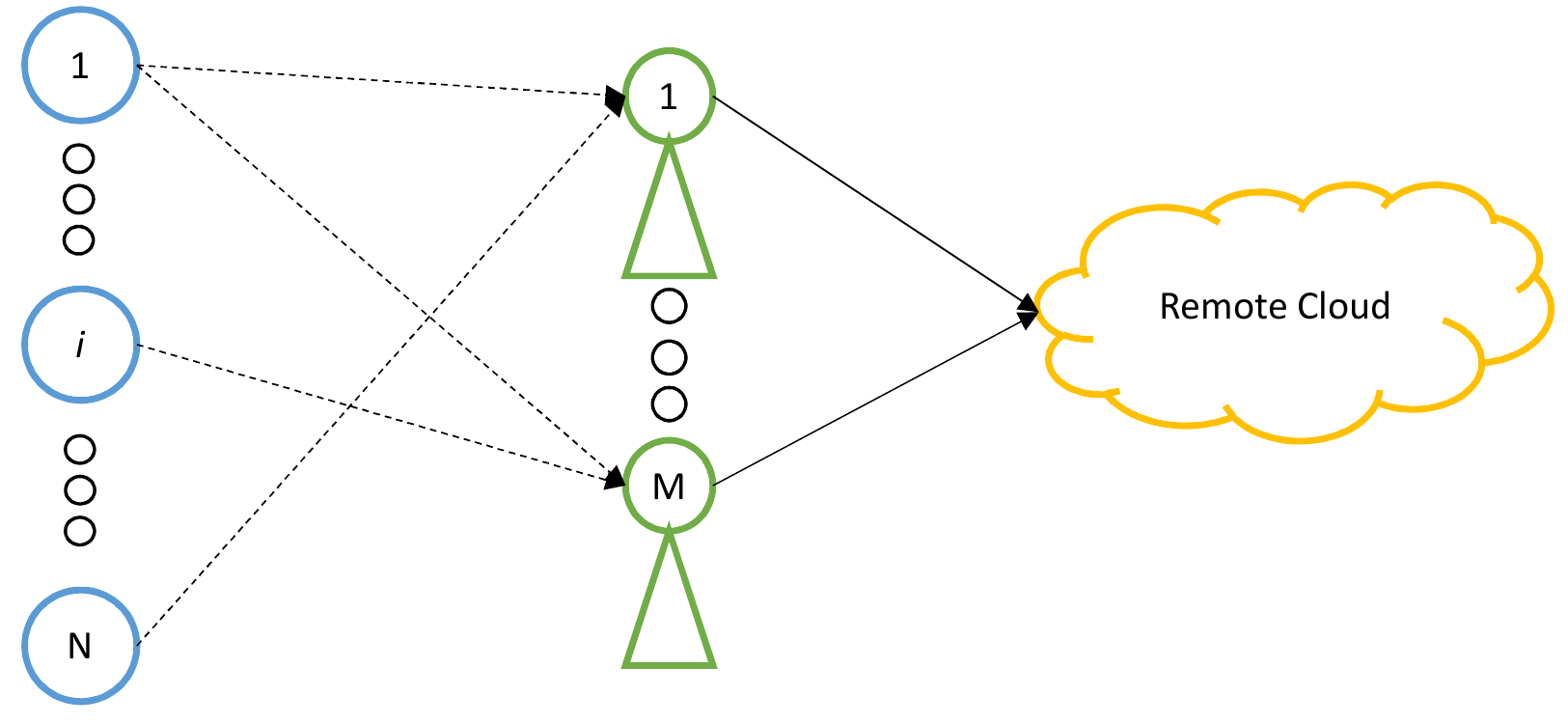}
        \caption{A mobile cloud system with CAPs.}
        \label{fig:system_model}
\end{figure}

\subsection{Offloading Decision}
In each round, a user may process their task locally, or offload 
it 
to one of the $M$ CAPs.  Each user $i$ has a possibly empty set of 
$\mathcal{K}_i \subset \mathcal{M}$ CAPs to which they may not offload their 
task to.  If a task is offloaded to a CAP, it may be processed there 
or further offloaded to the remote cloud server.  Denote the offloading 
decision of user $i$ by
\begin{align}
\mathbf{x}_i = [x_{i0}, x_{i1}, \ldots, x_{iM}] &\in \{0,1\}^{M+1}, 
\label{eq:x_int}\\
\theta_i &\in \{0,1\}, \label{eq:theta_int}
\end{align}
where $x_{i0} = 1$ if its task is locally processed locally at the user's own 
device, $x_{ij} = 1, j > 0$ if 
the task is offloaded to CAP $j$ , and $\theta_i = 1$ if the task 
that is offloaded to a CAP is processed at the cloud. We assume that the tasks 
are atomic and cannot be placed in multiple locations, only one of $x_{i0}, 
x_{i1}, \ldots, x_{iM}$ can be active --- this can be expressed through the 
following constraint:
\begin{equation}
\sum_{j=0}^{M}x_{ij} = 1, \forall i\in\mathcal{N}. \label{eq:x_sum}
\end{equation}
Additionally, processing at the cloud can only happen through transmission from 
one of the CAPs---thus $\theta_i \neq 1$ if $x_{i0} = 1$.  Alternately,
\begin{equation}
        \theta_i \leq \sum_{j=1}^{M}x_{ij}, \quad\forall i\in\mathcal{N}. 
        \label{eq:theta_sum}
\end{equation}

Finally, the placement constraints of each user are expressed as follows:
\begin{equation}
        x_{ik} = 0, \quad\forall k \in \mathcal{K}_i, i \in \mathcal{N}. 
        \label{eq:placement}
\end{equation}
\subsection{Cost of Local Processing}
The input data size, output data size, and required number of processing cycles 
of user
$i$'s task are denoted by $D_{{\text{in}}}(i)$,
$D_{{\textrm{\text{out}}}}(i)$, and $Y(i)$, respectively. For local 
computation, the 
processing time required is $T_{i}^L$, and the energy consumption for the user 
is $E_{i}^L$. 
\subsection{Cost of CAP Processing}
For tasks processed at a CAP, the bandwidth and processing rate that is 
allocated to each user must be determined in order to compute the energy and 
time required to complete a task.  Each CAP has its own wireless channel where 
users may 
offload their tasks to.  The uplink transmission time is given as
\begin{equation}
T_{ij}^t = \frac{D_{\text{in}}(i)}{\eta_{ij}^uc_{ij}^u},
\end{equation}
where $c_{ij}^u$ is the uplink 
bandwidth to CAP $j$ assigned to user $i$, and $\eta_{ij}^u$ is the spectral 
efficiency of the uplink transmission from user $i$ to CAP $j$.\footnote{The 
        spectral efficiency can be approximated by the Shannon bound 
        $\log(1+\text{SNR})$, where SNR 
        is 
        the signal-to-noise ratio between the user and the CAP.} 
        The 
downlink 
transmission time is given as
\begin{equation}
        T_{ij}^r = \frac{D_{\text{out}}(i)}{\eta_{ij}^dc_{ij}^d},
\end{equation}
where $c_{ij}^d$ and $ 
\eta_{ij}^d$ is the respective downlink bandwidth and spectral efficiency.  The 
bandwidths $c_{ij}^u$ and $c_{ij}^d$ are 
further limited by bandwidth capacities as follows:
\begin{align}
&\sum_{i=1}^{N}c_{ij}^u \leq C_j^{\text{UL}}, \quad\forall j\in\mathcal{M}, 
\label{eq:C_UL}\\
&\sum_{i=1}^{N}c_{ij}^d \leq C_j^{\text{DL}}, \quad\forall j\in\mathcal{M}, \\
&\sum_{i=1}^{N}(c_{ij}^u + c_{ij}^d) \leq C_j^{\text{Total}}, \quad\forall 
j\in\mathcal{M}.
\end{align}

Furthermore, a task is processed at the CAP with processing time $T_{ij}^a 
= Y(i)/f_{ij}^a$, where $f_{ij}^a$ is the processing rate at CAP $j$ assigned 
to user $i$, and constrained by the total processing rate of the CAP:
\begin{align}
\sum_{i=1}^{N}f_{ij}^a \leq f_j^A, \quad\forall j\in\mathcal{M}. \label{eq:f_A}
\end{align}

The total time required for processing task $i$ on CAP $j$ can then be 
expressed as
\begin{equation}
T_{ij}^A = T_{ij}^t + T_{ij}^r + T_{ij}^a.
\end{equation}

The energy consumption required for CAP processing can be expressed as
\begin{equation}
E_{ij}^A = E_{ij}^t + E_{ij}^r,
\end{equation}
where $E_{ij}^t$ and $E_{ij}^r$ represent the uplink and downlink transmission 
energy costs required by the user, respectively.

\subsection{Cost of Cloud Processing}
For cloud processing, the uplink and downlink transmission energies and times 
are identical to the CAP processing case.  However, there is now further 
transmission time between the CAP and the cloud, expressed as $T_i^{ac} = 
(D_{\text{in}}(i) + D_{\text{out}}(i))/r_{ac}$, where $r_{ac}$ is a 
predetermined 
wired 
transmission rate between the CAPs and the cloud.  The computation time is 
$T_{i}^c = Y(i)/f_i^C$, where $f_i^C$ is a predetermined processing 
rate 
assigned to each user at the cloud.  Thus, the total cloud processing time can 
be expressed as
\begin{equation}
T_{ij}^C = T_{ij}^t + T_{ij}^r + T_i^{ac} + T_{i}^c.
\end{equation} 

The energy consumption can be expressed as
\begin{equation}
E_{ij}^C = E_{ij}^t + E_{ij}^r + \beta_iC_i^C,
\end{equation}
where $C_{C_i}$ is the cloud utility cost, and $\beta_i$ is the relative weight.

\begin{table}[t]
\caption{Parameters and Respective Descriptions.} \centering

\begin{tabular}{|r| l|}
\hline 
\textbf{Parameter}& \textbf{Description} \\
\hline 
$E_{i}^L$ & local processing energy of user $i$'s task\\
$E_{ij}^t$,\ $E_{ij}^r$ & uplink transmitting energy and downlink \\ 
                & receiving energy of user $i$'s task to and \\ 
                & from CAP $j$, respectively\\
$T_{i}^L$,\ $T_{i}^{C}$ & local processing time and cloud processing \\
                &time of user $i$'s task, respectively \\
$T_{ij}^t$,\ $T_{ij}^t$ & uplink transmission time and downlink \\              
                                        &       transmission time of user $i$'s task between \\
                        & a mobile user and CAP $j$, respectively\\
$T_i^{ac}$ & transmission time of user $i$'s task between \\
         &CAP and cloud\\
$C^{\text{UL}}_j$,\ $C^{\text{DL}}_j$,\ $C^{\text{Total}}_j$ & uplink, 
downlink, and total 
transmission \\
            &capacity  associated with CAP $j$, respectively \\
$c_{ij}^u$,\ $c_{ij}^d$ & uplink and downlink bandwidth\\
         & assigned to user $i$ from CAP $j$\\
$f_{ij}^a$ & CAP processing rate assigned to user $i$ from \\ & CAP $j$\\
$\eta_{ij}^u$,\ $\eta_{ij}^d$ & uplink and downlink spectral efficiency\\
& from user $i$ to CAP $j$\\         
$C_i^C$ & system utility cost of user $i$'s task\\
$r_{ac}$ & transmission rate for each user\\ & between the CAP and cloud \\
$f_i^{C}$ & cloud processing rate for each user\\
 \hline
\end{tabular}
\label{table_notation}
\vspace{-0.5cm}
\end{table}

These various parameters are summarized in Table \ref{table_notation}.

\subsection{Optimization Problem Formulation}
To maintain adequate quality of service for each user, we seek to minimize 
both the energy consumption and the processing time.  
Because of the round-by-round schedule, each user experiences the same 
processing delay between tasks, which is the total round time.
Thus, our centralized optimization problem considers a weighted sum of the 
total energy consumption among the users and the round time, which is the 
maximum processing time among all of the tasks in the round.  Our decision 
variables are the 
offloading decisions $\mathbf{x}_i$ and $\theta_i$, as well as the resource 
allocation decisions $\mathbf{c}_i^u = [c_{i1}^u, \ldots, c_{iM}^u], 
\mathbf{c}_i^d = [c_{i1}^d, \ldots, c_{iM}^d]$, and $\mathbf{f}_i^a = 
[f_{i1}^a, \ldots, f_{iM}^a]$.  Thus, the optimization problem is as follows:

\begin{align}
\min_{\{\mathbf{x}_i\}, \{\theta_i\} \{\mathbf{c}_i^u\}, \{\mathbf{c}_i^d\}, 
        \{\mathbf{f}_i^a\}} \Bigg\{&\sum_{i=1}^{N}\alpha_i(E_{L_i} + E_{A_i} + 
E_{C_i}) \nonumber \\
&+ \max_{i'\in\mathcal{N}}\{T_{L_{i'}} + T_{A_{i'}} + T_{C_{i'}}\}\Bigg\},  
\label{eq:global_obj}\\
\text{subject to } &\text{(\ref{eq:x_int})-(\ref{eq:placement}), 
(\ref{eq:C_UL})-(\ref{eq:f_A})}, \nonumber \\
c_{ij}^u, c_{ij}^d, f_{ij}^a \geq 0, &\quad\forall i \in \mathcal{N}, j \in 
\mathcal{M}, \label{eq:geq0}
\end{align}
where
\begin{align*}
E_{L_i} &= E_i^Lx_{i0}, &T_{L_i} &= T_i^lx_{i0}, \\
E_{A_i} &= \sum_{j=1}^{M}E_{ij}^Ax_{ij}(1-\theta_i), &T_{A_i} &= 
\sum_{j=1}^{M}T_{ij}^Ax_{ij}(1-\theta_i), \\
E_{C_i} &= \sum_{j=1}^{M}E_{ij}^Cx_{ij}\theta_i, &T_{C_i} &= 
\sum_{j=1}^{M}T_{ij}^Cx_{ij}\theta_i, \\
&\quad i \in \mathcal{N},
\end{align*}

where $\alpha_i$ is a relative weight between the energy consumption and the 
round time.

\subsection{Selfish User Assumption}
Since the optimization problem (\ref{eq:global_obj}) does not address the 
actions of 
selfish users, we must also consider the individual objective function, defined 
as a weighted sum of individual energy consumption and round time, 
\begin{align}
u_i(a_1, \ldots, a_N) &= \alpha_i(E_{L_i} + E_{A_i} + E_{C_i}) \nonumber \\ 
&+ \max_{i'\in\mathcal{N}}\{T_{L_{i'}} + T_{A_{i'}} + T_{C_{i'}}\}, 
\label{eq:ind_cost}
\end{align}
where $a_i = (\mathbf{x}_i, \theta_i)$ represents the offloading decision of 
user $i$.

Thus, the users may decide against choosing the offloading decision dictated by 
some centralized solution to the problem above.  To account for this, we must 
arrive at a solution that is an NE, where no user has any 
incentive to deviate from their dictated solution given the behaviours of the 
other users.


\section{MCAP Offloading Solution}

The resultant centralized optimization problem is a non-convex mixed integer 
programming problem, and no known method exists for efficiently finding a 
global optimum.  We 
therefore propose the MCAP method for finding a heuristic solution.  In 
this method, we transform the original problem into a quadratically-constrained 
quadratic programming (QCQP) problem, then relax the integer constraints on 
the offloading variables in order to produce an SDP over 
real-valued offloading decisions variables, as explicated 
in \cite{luo2010}.  This overall QCQP-SDP solution framework is generic and was 
used also in \cite{chen2016icassp} and \cite{chen2017} for the case of a single 
CAP. However,
because of the presence of multiple CAPs, there are 
new challenges in applying these methods to the current system, due to 
the substantial addition of objective variables and constraints that result.  
In particular, we must accommodate for the additional offloading variable 
$\theta$ and associated constraints (\ref{eq:theta_int}) and 
(\ref{eq:theta_sum}), the additional dimensionality of offloading variables 
$\mathbf{x}_i$ and resource allocation variables $c_{ij}^u$, $c_{ij}^d$, and 
$f_{ij}^a$, as well as their associated constraints 
(\ref{eq:C_UL})-(\ref{eq:f_A}), and the addition of placement constraints 
(\ref{eq:placement}).
   
Through the solution of the SDR, a probabilistic 
mapping method is used to recover a set of binary-valued offloading 
decisions.  Once a set of offloading decisions have been recovered, we compute 
the optimal resource allocation problem for each one, which can 
be expressed as the following optimization problem:

\begin{align}
\min_{\{\mathbf{c}_i^u\}, \{\mathbf{c}_i^d\}, 
        \{\mathbf{f}_i^a\}} \Bigg\{&\sum_{i=1}^{N}\alpha_i(E_{L_i} + E_{A_i} + 
E_{C_i}) \nonumber \\
&+ \max_{i'\in\mathcal{N}}\{T_{L_{i'}} + T_{A_{i'}} + T_{C_{i'}}\}\Bigg\},  
\label{eq:resource_alloc} \\
&\text{s.t. (\ref{eq:C_UL})-(\ref{eq:f_A}), (\ref{eq:geq0}), } \nonumber
\end{align}
which is a convex problem, since it is a maximum of sums of positive reciprocal 
functions with linear constraints, and thus can be solved using any standard 
convex solver.  From these solutions, we choose the one that minimizes the 
system objective (\ref{eq:global_obj}).

\subsection{QCQP Formulation}
In order to reformulate (\ref{eq:global_obj}) as a QCQP, we first replace the 
integer 
constraints (\ref{eq:x_int}) and (\ref{eq:theta_int}) with the following 
quadratic constraints:
\begin{align}
        x_{ij}(x_{ij} - 1) &= 0, \quad\forall i\in\mathcal{N}, j\in 
        \{0\}\cup\mathcal{M} \label{eq:x_quad}
        \\
        \theta_i(\theta_i - 1) &= 0, \quad\forall i\in\mathcal{N}. 
        \label{eq:theta_quad}
\end{align}

We then introduce an auxiliary variable $t$ in order to remove the maximum in 
the objective and re-express that condition as a constraint.  The objective now 
becomes
\begin{align}
        \min_{\{\mathbf{x}_i\}, \{\mathbf{c}_i^u\}, \{\mathbf{c}_i^d\}, 
                \{\mathbf{f}_i^a\}, t} \Bigg\{&\sum_{i=1}^{N}\alpha_i(E_{L_i} + 
                E_{A_i}, 
                + 
        E_{C_i}) + t \Bigg\},
\end{align}
with a new delay constraint:
\begin{align}
        T_i^Lx_{i0} &+ 
        \sum_{j=1}^{M}\Bigg(\frac{D_{\text{in}}(i)}{\eta_{ij}^uc_{ij}^u} + 
        \frac{D_{\text{out}}(i)}{\eta_{ij}^dc_{ij}^d} + 
        \frac{Y(i)}{f_{ij}^a}(1-\theta_i)\Bigg)x_{ij} \nonumber \\
         &+ \big(T_i^{ac} + T_i^c\big)\theta_i \leq t, \quad\forall i \in 
         \mathcal{N}.
\end{align}
We now introduce a set of auxiliary variables $D_{ij}^u, D_{ij}^d,$ and 
$D_{ij}^a$ 
for each set of users and CAPs, and replace the above delay constraint with the 
following:
\begin{align}
        \frac{D_{\text{in}}(i)}{\eta_{ij}^uc_{ij}^u}x_{ij} &\leq D_{ij}^u, 
        \quad\forall 
        i \in 
        \mathcal{N}, j \in \mathcal{M}, \label{eq:D_u} \\
        \frac{D_{\text{out}}(i)}{\eta_{ij}^dc_{ij}^d}x_{ij} &\leq D_{ij}^d, 
        \quad\forall 
        i \in 
        \mathcal{N}, j \in \mathcal{M}, \\
        \frac{Y(i)}{f_{ij}^a}x_{ij}(1-\theta_i) &\leq D_{ij}^a, \quad\forall i \in 
        \mathcal{N}, j \in \mathcal{M}, \label{eq:D_a} \\
        T_i^Lx_{i0} &+ \sum_{j=1}^{M}(D_{ij}^u + D_{ij}^d + D_{ij}^a) \nonumber \\
        &+ (T_i^{ac} +  T_i^C)\theta_i \leq t, \quad\forall i \in \mathcal{N}. 
        \label{eq:delay}
\end{align}
With $\mathbf{D}_i^u = [D_{i0}^u,\ldots, D_{iM}^u]$, $\mathbf{D}_i^d = 
[D_{i0}^d,\ldots, D_{iM}^d]$ and $\mathbf{D}_i^a = [D_{i0}^a,\ldots, D_{iM}^a]$
we can then vectorize all of the variables and parameters into one vector 
$\mathbf{w}$, defined as
\begin{align}
        \mathbf{w} = [&\mathbf{x}_1,\ldots,\mathbf{x}_N,\theta_1,\ldots,\theta_N, 
        \mathbf{c}_1^u,\ldots,\mathbf{c}_N^u,\mathbf{D}_1^u,\ldots,\mathbf{D}_N^u, 
        \nonumber \\ 
        &\mathbf{c}_1^d,\ldots,\mathbf{c}_N^d,\mathbf{D}_1^d,\ldots,\mathbf{D}_N^d
        \mathbf{f}_1^a,\ldots,\mathbf{f}_N^a,\mathbf{D}_1^a,\ldots,\mathbf{D}_N^a,t]^T
        \nonumber \\
        &\in \mathbb{R}^{7MN+2N+1}.
\end{align}
By doing so, we can rewrite our objective function as
\begin{align}
        \mathbf{b}_0^T\mathbf{w}, \label{eq:obj_w}
\end{align}
where
\begin{align*}
        \mathbf{b}_0 = 
        [&\alpha_1E_1^l,\alpha_1E_{1,1}^A,\ldots,\alpha_1E_{1,M}^A,\ldots,\alpha_NE_{N,M}^A,
         \nonumber \\
        &\alpha_1\beta_1C_{C_1},\ldots,\alpha_1\beta_1C_{C_1}, 
        \mathbf{0}_{1\times6MN}, 1]^T.
\end{align*}

Similarly, each of the constraints above can be rewritten in matrix form.  The 
time constraint (\ref{eq:delay}) can be expressed as
\begin{align}
        (\mathbf{b}_{i}^c)^T\mathbf{w} \leq 0, \quad\forall i \in \mathcal{N}, 
        \label{eq:QCQPtime}
\end{align}
where
\begin{align*}
        \mathbf{b}_{i}^c= [&\mathbf{0}_{1\times(M+1)(i-1)+1}, T_i^L, \nonumber \\
        &\mathbf{0}_{1\times(M+1)(N-i+M-1)},\mathbf{0}_{1\times(i-1)},
        T_i^{ac}+T_i^C, \nonumber \\ 
        &\mathbf{0}_{1\times(N-i)}\mathbf{0}_{1\times(MN+(i-1)N)},
        \mathbf{1}_{1\times{}N},\mathbf{0}_{1\times(N-i)M}, \nonumber \\
        &\mathbf{0}_{1\times(N-i)}\mathbf{0}_{1\times(MN+(i-1)N)},
        \mathbf{1}_{1\times{}N},\mathbf{0}_{1\times(N-i)M}, \nonumber \\
        &\mathbf{0}_{1\times(N-i)}\mathbf{0}_{1\times(MN+(i-1)N)},
        \mathbf{1}_{1\times{}N},\mathbf{0}_{1\times(N-i)M}, -1]^T.
\end{align*}
To express constraints (\ref{eq:D_u})-(\ref{eq:D_a}), we rearrange their 
respective 
expressions to produce the following equivalent inequalities:
\begin{align*}
D_{\text{in}}(i)x_{ij} - \eta_{ij}^uc_{ij}^uD_{ij}^u &\leq 0, \\
D_{\text{out}}(i)x_{ij} - \eta_{ij}^uc_{ij}^uD_{ij}^d &\leq 0, \\
Y(i)x_{ij} - Y(i)x_{ij}\theta_i - f_{ij}^aD_{ij}^a &\leq 0 .\\
\end{align*}
Equations (\ref{eq:D_u})-(\ref{eq:D_a}) can now be expressed as
\begin{align}
        &\mathbf{w}^TA_{ij}^\mu{}\mathbf{w} + (\mathbf{b}_i^\mu)^T\mathbf{w} \leq 
        0, \mu\in\{u,d,a\}, \label{eq:mu}\\ 
        &\quad\forall i\in \mathcal{N}, j\in\mathcal{M}, \nonumber
\end{align}
where $\mathbf{e}_i \text{ is a unit vector of size $MN$ of dimension $i$},$ and
\scriptsize

\begin{align*}
&\mathbf{A}_{ij}^{u'} = -0.5\eta_{ij}^u\begin{bmatrix}
\mathbf{0}_{MN\times{}MN} & \text{diag}(\mathbf{e}_i) \\
\text{diag}(\mathbf{e}_i) & \mathbf{0}_{MN\times{}MN}
\end{bmatrix}, \\
&\mathbf{A}_{ij}^{d'} = -0.5\eta_{ij}^d\begin{bmatrix}
\mathbf{0}_{MN\times{}MN} & \text{diag}(\mathbf{e}_i) \\
\text{diag}(\mathbf{e}_i) & \mathbf{0}_{MN\times{}MN}
\end{bmatrix}, \\
&\mathbf{A}_{ij}^{a_1'} = -0.5\begin{bmatrix}
\mathbf{0}_{MN\times{}MN} & \text{diag}(\mathbf{e}_i) \\
\text{diag}(\mathbf{e}_i) & \mathbf{0}_{MN\times{}MN}
\end{bmatrix}, \\
&\mathbf{A}_{ij}^{a_2'} \in \mathbb{R}^{(M+1)N\times(M+1)N} \text{such that} \\
&\Big\{\mathbf{A}_{ij}^{a_2'}\Big\}_{\alpha,\beta} = \begin{cases}
-0.5Y(i), & \alpha = (M+1)N+j, \beta = (M+1)(i-1)+j \\ 
-0.5Y(i), & \alpha = (M+1)(i-1)+j, \beta = (M+1)N+j \\
0, & \text{otherwise}
\end{cases},\\
&\mathbf{A}_{ij}^{u} = \begin{bmatrix}
\mathbf{0}_{(M+2)N\times(M+2)N} & \mathbf{0}_{(M+2)N\times{}2MN} & 
\mathbf{0}_{(M+2)N\times{}4MN+1} \\
\mathbf{0}_{2MN\times(M+2)N} & \mathbf{A}_{ij}^{u'} & 
\mathbf{0}_{2MN\times{}4MN+1}, \\
\mathbf{0}_{4MN+1\times(M+2)N} & \mathbf{0}_{4MN+1\times{}2MN} & 
\mathbf{0}_{4MN+1\times{}4MN+1}
\end{bmatrix}, \\
&\mathbf{A}_{ij}^{d} = \begin{bmatrix}
\mathbf{0}_{(3M+2)N\times(3M+2)N} & \mathbf{0}_{(3M+2)N\times{}2MN} & 
\mathbf{0}_{(3M+2)N\times{}2MN+1} \\
\mathbf{0}_{2MN\times(3M+2)N} & \mathbf{A}_{ij}^{d'} & 
\mathbf{0}_{2MN\times{}2MN+1} \\
\mathbf{0}_{2MN+1\times(3M+2)N} & \mathbf{0}_{2MN+1\times{}2MN} & 
\mathbf{0}_{2MN+1\times{}2MN+1}
\end{bmatrix}, \\
&\mathbf{A}_{ij}^{a_1} = \begin{bmatrix}
\mathbf{0}_{(5M+2)N\times(5M+2)N} & \mathbf{0}_{(5M+2)N\times{}2MN} & 
\mathbf{0}_{(5M+2)N\times{}1} \\
\mathbf{0}_{2MN\times(5M+2)N} & \mathbf{A}_{ij}^{a_1'} & 
\mathbf{0}_{2MN\times{}1} \\
\mathbf{0}_{1\times(5M+2)N} & \mathbf{0}_{1\times{}2MN} & 0
\end{bmatrix}, \\
&\mathbf{A}_{ij}^{a_2} = \begin{bmatrix}
\mathbf{A}_{ij}^{a_2'} & \mathbf{0}_{(M+2)N\times{}6MN+1} \\
\mathbf{0}_{6MN+1\times{}(M+2)N} & \mathbf{0}_{6MN+1\times{}6MN+1}
\end{bmatrix}, \\
&\mathbf{A}_{ij}^{a} = \mathbf{A}_{ij}^{a_1} + \mathbf{A}_{ij}^{a_2}, \\
&\mathbf{b}_{ij}^u = [\mathbf{0}_{1\times{}(M+1)(i-1)+j-1}, D_{\text{in}}(i), 
\mathbf{0}_{1\times(7N-i+1)M+2N-i-j}]^T, \\
&\mathbf{b}_{ij}^d = [\mathbf{0}_{1\times{}(M+1)(i-1)+j-1}, D_{\text{out}}(i), 
\mathbf{0}_{1\times(7N-i+1)M+2N-i-j}]^T, \\
&\mathbf{b}_{ij}^a = [\mathbf{0}_{1\times{}1\times{}(M+1)(i-1)+j-1}, Y(i), 
\mathbf{0}_{1\times(7N-i+1)M+2N-i-j}]^T. \\
\end{align*}
\normalsize
The offloading constraints (\ref{eq:x_sum}) and (\ref{eq:theta_sum}) become
\begin{align}
        \mathbf{b}_i^P\mathbf{w} &= 1, \quad\forall i \in \mathcal{N}, 
        \label{eq:w_p}\\
        \mathbf{b}_i^Q\mathbf{w} &\leq 0, \quad\forall i \in \mathcal{N}, 
        \label{eq:w_q}
\end{align}
where
\begin{align*}
\mathbf{b}_i^P &= [\mathbf{0}_{1\times(i-1)(M+1)}, \mathbf{1}_{1\times{}M+1}, 
\mathbf{0}_{1\times{}7MN+2N-iM-i+1}]^T, \\
\mathbf{b}_i^Q &= [\mathbf{0}_{1\times(i-1)(M+1)+1} -\mathbf{1}_{1\times{}M}, 
\\ &\mathbf{0}_{1\times{}(N-i)(M+1)+i-1}, 1, \mathbf{0}_{6MN+1-N+i}]^T.
\end{align*}
The bandwidth and processing capacities constraints 
(\ref{eq:C_UL})-(\ref{eq:f_A}) become
\begin{align}
\mathbf{b}_j^U\mathbf{w} &\leq C_j^{\text{UL}}, \quad\forall j \in \mathcal{M}, 
\label{w_u} 
\\
\mathbf{b}_j^D\mathbf{w} &\leq C_j^{\text{DL}}, \quad\forall j \in \mathcal{M}, 
\\
\mathbf{b}_j^S\mathbf{w} &\leq C_j^{\text{Total}}, \quad\forall j \in 
\mathcal{M}, \\
\mathbf{b}_j^A\mathbf{w} &\leq f_j^A, \quad\forall j \in \mathcal{M}, 
\label{w_a}
\end{align}
where
\begin{align*}
\mathbf{b}_j' &= [\mathbf{0}_{1\times{}j-1}, 1, \mathbf{0}_{1\times{}M-j}], \\
\mathbf{b}_j^U &= [\mathbf{0}_{1\times{}(M+2)N}, 
\{\mathbf{b}_j'\}_{1\times{}N}, \mathbf{0}_{1\times{}5MN+1}], \\
\mathbf{b}_j^D &= [\mathbf{0}_{1\times{}(3M+2)N}, 
\{\mathbf{b}_j'\}_{1\times{}N}, \mathbf{0}_{1\times{}3MN+1}], \\
\mathbf{b}_j^S &= \mathbf{b}_j^U + \mathbf{b}_j^D, \\
\mathbf{b}_j^A &= [\mathbf{0}_{1\times{}(5M+2)N}, 
\{\mathbf{b}_j'\}_{1\times{}N}, \mathbf{0}_{1\times{}MN+1}]. \\
\end{align*}
The nonnegative constraint (\ref{eq:geq0}) becomes
\begin{align}
        \mathbf{w} \succeq 0, \label{succeq0}
\end{align}
while the integer constraints (\ref{eq:x_quad})-(\ref{eq:theta_quad}) can be 
written as
\begin{align}
        \mathbf{w}^T\mathbf{e}_p\mathbf{w} - \mathbf{e}_p^T\mathbf{w} = 0, 
        \quad\forall p \in 
        \{1,\ldots,(M+2)N\}, \label{w_p}
\end{align}
where $\mathbf{e}_p$ is a unit vector of size $7MN+2N+1$ of dimension $p$.

Finally, the placement constraints (\ref{eq:placement}) can be expressed 
as\footnote{Strictly 
speaking, there should be one equality constraint for each individual placement 
constraint of each user.  However, because we have a nonnegative constraint 
(\ref{succeq0}), the sum constraint provided above is sufficient.}
\begin{align}
\mathbf{b}_i^K\mathbf{w} = 0, \quad\forall i\in\mathcal{N}, 
\label{eq:QCQPplacement}
\end{align}
where
\begin{align*}
\mathbf{b}_i^K = \sum_{k\in\mathcal{K}_i} \mathbf{e}_{(M+2)(i-1)+k}. 
\end{align*}

Hence, the QCQP formulation of (\ref{eq:global_obj}) can be expressed 
equivalently as
\begin{align}
        \min_{\mathbf{w}} \mathbf{b}_0\mathbf{w}, \label{eq:QCQP} \\
        \text{s.t. (\ref{eq:QCQPtime})-(\ref{eq:QCQPplacement})}. \nonumber
\end{align}

\subsection{SDP Solution}
In order to arrive at the SDP formulation from (\ref{eq:QCQP}), we must express 
the problem in matrix form.  We begin by defining the vector $z = [\mathbf{w} 
\: 1]^T$, and reformulate (\ref{eq:QCQP}) in terms of $\mathbf{z}$.  We denote 
$\mathbf{0}$ as the zero matrix of dimension $7MN+2N+1\times7MN+2N+1$, and 
$\mathbf{e}_p$ is a unit vector of size $7MN+2N+1$ and dimension $p$.

The objective function from (\ref{eq:obj_w}) now becomes
\begin{align}
        \min_{\mathbf{z}} \mathbf{z}^T\mathbf{G}_0\mathbf{z},
\end{align}
where 
\begin{align*}
        \mathbf{G}_0 = \begin{bmatrix}
\mathbf{0} & \frac{1}{2}\mathbf{b}_0 \\
\frac{1}{2}\mathbf{b}_0^T & 0
\end{bmatrix}.
\end{align*}
Constraint (\ref{eq:QCQPtime}) is now expressed as
\begin{align}
\mathbf{z}^T\mathbf{G}_{i}^c\mathbf{z} \leq 0, \quad\forall i 
\in \mathcal{N}, \label{z_c}
\end{align}
where 
\begin{align*}
        \mathbf{G}_i^c &= \begin{bmatrix}
\mathbf{0} & \frac{1}{2}\mathbf{b}_i^c \\
\frac{1}{2}(\mathbf{b}_i^c)^T & 0
\end{bmatrix};
\end{align*}
(\ref{eq:mu}) is expressed as
\begin{align}
\mathbf{z}^T\mathbf{G}_{ij}^\mu{}\mathbf{z} \leq 0, \quad\forall \mu \in 
\{u,d,a\}, i \in 
\mathcal{N}, j \in \mathcal{M}
\end{align}
where
\begin{align*}
        \mathbf{G}_{ij}^\mu &= \begin{bmatrix}
\mathbf{A}_{ij}^\mu &  \frac{1}{2}\mathbf{b}_{ij}^\mu \\
\frac{1}{2}(\mathbf{b}_{ij}^\mu)^T & 0
\end{bmatrix}, \quad \mu \in 
\{u,d,a\};
\end{align*}
(\ref{eq:w_p}) and (\ref{eq:w_q}) are expressed as
\begin{align}
        &\mathbf{z}^T\mathbf{G}_i^P\mathbf{z} = 1, \quad\forall i \in \mathcal{N}, 
\\
&\mathbf{z}^T\mathbf{G}_i^Q\mathbf{z} \leq 0, \quad\forall i \in 
\mathcal{N},
\end{align}
where
\begin{align*}
        \mathbf{G}_i^P &= \begin{bmatrix}
\mathbf{0} & \frac{1}{2}\mathbf{b}_i^P \\
\frac{1}{2}(\mathbf{b}_i^P)^T & 0
\end{bmatrix}, \\
\mathbf{G}_i^Q &= \begin{bmatrix}
\mathbf{0} & \frac{1}{2}\mathbf{b}_i^Q \\
\frac{1}{2}(\mathbf{b}_i^Q)^T & 0
\end{bmatrix};
\end{align*}
(\ref{w_u})-(\ref{w_a}) are expressed as
\begin{align}
        &\mathbf{z}^T\mathbf{G}_j^U\mathbf{z} \leq C_j^{\text{UL}}, \quad\forall j 
\in 
\mathcal{M}, \\
&\mathbf{z}^T\mathbf{G}_j^D\mathbf{z} \leq C_j^{\text{DL}}, \quad\forall j 
\in 
\mathcal{M}, \\
&\mathbf{z}^T\mathbf{G}_j^S\mathbf{z} \leq C_j^{\text{Total}}, \quad\forall 
j 
\in 
\mathcal{M}, \\
&\mathbf{z}^T\mathbf{G}_j^A\mathbf{z} \leq f_j^A, \quad\forall j \in 
\mathcal{M}, 
\end{align}
where
\begin{align*}
        \mathbf{G}_j^\pi &= \begin{bmatrix}
\mathbf{0} & \frac{1}{2}\mathbf{b}_j^\pi \\
\frac{1}{2}(\mathbf{b}_j^\pi)^T & 0
\end{bmatrix}, \pi \in \{U,D,S,A\};
\end{align*}
(\ref{succeq0}) is expressed as
\begin{align}
\mathbf{z} \succeq 0;
\end{align}
(\ref{w_p}) is expressed as
\begin{align}
        &\mathbf{z}^T\mathbf{G}_p^I\mathbf{z} = 0, \quad\forall p \in 
\{1,\ldots,(M+2)N\},
\end{align}
where
\begin{align*}
        \mathbf{G}_p^I &= \begin{bmatrix}
\text{diag}(\mathbf{e}_p) & -\frac{1}{2}\mathbf{e}_p \\
-\frac{1}{2}(\mathbf{e}_p)^T & 0
\end{bmatrix};
\end{align*}
and finally, (\ref{eq:QCQPplacement}) is expressed as
\begin{align}
        &\mathbf{z}^T\mathbf{G}_i^K\mathbf{z} = 0, \quad\forall i \in \mathcal{N}, 
        \label{z_K} 
\end{align}
where
\begin{align*}
        \mathbf{G}_i^K &= \begin{bmatrix}
\mathbf{0} & \frac{1}{2}\mathbf{b}_i^K \\
\frac{1}{2}(\mathbf{b}_i^K)^T & 0
\end{bmatrix}.
\end{align*}

Problem (\ref{eq:QCQP}) can now be equivalently transformed to:
\begin{align}
\min_{\mathbf{z}} \mathbf{z}^T\mathbf{G}_0\mathbf{z}, \label{eq:z_obj}\\
\text{s.t. (\ref{z_c})-(\ref{z_K})}. \nonumber
\end{align}

Define $\mathbf{Z} = \mathbf{zz}^T$.  We can reformulate (\ref{eq:z_obj}) as an 
SDP under $\mathbf{Z}$, subject to an additional rank constraint: 
rank$(\mathbf{Z}) = 1$.  
Note that 
because $\mathbf{z} \succeq 0$, $\mathbf{Z}$ must be both positive semidefinite 
and 
element-wise nonnegative in order for the above problems to be equivalent. By 
dropping the rank constraint, we have the following convex problem:
\begin{align}
\min_{\mathbf{Z}} &\Tr(\mathbf{G}_0\mathbf{Z}), \label{eq:SDP} \\
\text{s.t. } &\Tr(\mathbf{G}_i^c\mathbf{Z}) \leq 0, \quad\forall i \in 
\mathcal{N}, \\
&\Tr(\mathbf{G}_{ij}^\mu{}\mathbf{Z}) \leq 0, \quad\forall \mu\in\{u,d,a\}, i 
\in 
\mathcal{N}, j \in \mathcal{M}, \\
&\Tr(\mathbf{G}_i^P\mathbf{Z}) = 1, \quad\forall i \in \mathcal{N}, 
\label{eq:SDP_p}\\
&\Tr(\mathbf{G}_i^Q\mathbf{Z}) \leq 0, \quad\forall i \in \mathcal{N}, 
\label{eq:SDP_q} \\
&\Tr(\mathbf{G}_j^U\mathbf{Z}) \leq C_j^{\text{UL}}, \quad\forall j \in 
\mathcal{M}, \\
&\Tr(\mathbf{G}_j^D\mathbf{Z}) \leq C_j^{\text{DL}}, \quad\forall j \in 
\mathcal{M}, \\
&\Tr(\mathbf{G}_j^S\mathbf{Z}) \leq C_j^{\text{Total}}, \quad\forall j \in 
\mathcal{M}, \\
&\Tr(\mathbf{G}_j^A\mathbf{Z}) \leq f_j^A, \quad\forall j \in \mathcal{M}, \\
&\Tr(\mathbf{G}_p^I\mathbf{Z}) = 0, \quad\forall p \in \{1,\ldots,(M+2)N\}, \\
&\Tr(\mathbf{G}_i^K\mathbf{Z}) = 0, \quad\forall i \in \mathcal{N}, 
\label{eq:SDP_k} \\
&\mathbf{Z} \succeq 0 \text{ (elementwise)}, \label{eq:SDPgeq0} \\
&\mathbf{Z} \in \mathbb{S}_{7MN+2N+2}^+.
\end{align}

This is a standard form SDP problem, and thus can be solved in polynomial time 
using an SDP solver.  Denote the resultant solution as $\mathbf{Z^*}$.  From 
this, we must recover a set of binary offloading decisions as a heuristic 
solution to the original problem (\ref{eq:global_obj}).  There are multiple 
means of recovering a rank-1 solution from a relaxed 
SDP, as detailed in \cite{luo2010}.  One such 
method involves generating a set of vectors from a zero-mean, $\mathbf{Z^*}$ 
covariance Gaussian distribution and mapping each element to the set of 
possible decisions $\{0,1\}$.  Such a method however would not guarantee that 
that constraints (\ref{eq:x_sum}) and (\ref{eq:theta_sum}) would be satisfied.  
Given the parameters of our 
particular problem, we instead use a randomization method to recover our 
solution, adopted in \cite{chen2016icassp}.

Note that the elements $\mathbf{Z}(7MN+2N+2,1)$ to 
$\mathbf{Z}(7MN+2N+2,(M+1)N)$ correspond to the CAP decision vectors 
$\mathbf{x}_1$ to $\mathbf{x}_N$, while $\mathbf{Z}(7MN+2N+2,(M+1)N+1)$ to 
$\mathbf{Z}(7MN+2N+2,(M+2)N)$ correspond to the cloud decision variables 
$\theta_1$ to $\theta_N$.  This arises from the fact that $\mathbf{Z} = 
\mathbf{zz}^T$ and $\mathbf{z} = [\mathbf{w} \: 1]$, and therefore the last row 
of $Z$ correspond to the elements in $\mathbf{w}$, the first $(M+2)N$ of which 
are equal to $[\mathbf{x}_1,\ldots,\mathbf{x}_N,\theta_1,\ldots,\theta_N]$. In 
order to use these results in our randomization method, we must prove the 
following lemma. Even though its is similar in form to Lemma 1 in 
\cite{chen2016icassp}, for the case of multiple CAPs we must reformulate the 
lemma 
for the dimensions of our $\mathbf{Z}$ matrix, as well as account for the 
additional offloading decision variable $\theta$ and associated constraint 
(\ref{eq:theta_sum}).
\begin{Lem}
        For the optimal solution $\mathbf{Z^*}$ of (\ref{eq:SDP}), 
        $\mathbf{Z^*}(7MN+2N+2,p) 
        \in [0,1], \forall p \in \{1,\ldots,(M+2)N\}$
\end{Lem}
\begin{proof}
From (\ref{eq:SDPgeq0}), we are guaranteed that all elements of $\mathbf{Z^*}$ 
are 
nonnegative.  For $p \in \{1,\ldots,(M+1)N\}$, constraint (\ref{eq:SDP_p}) 
requires that
\begin{align*}
\sum_{p=(M+1)(i-1)+1}^{(M+1)i}&\mathbf{Z^*}(7MN+2N+2,p) = 1\\
 &\quad\forall i \in \{1,\ldots,N\},
\end{align*}
which ensures that each of the above elements individually are less than or 
equal to 
1.  For $p \in \{(M+1)N+1,\ldots,(M+2)N\}$, (\ref{eq:SDP_q}) requires that
\begin{align*}
\mathbf{Z^*}&(7MN+2N+2,(M+1)N+i) \leq \\ 
\sum_{p=(M+1)(i-1)+2}^{(M+1)i}&\mathbf{Z^*}(7MN+2N+2,p), \\
&\quad \forall i \in \{1,\ldots,N\}.
\end{align*}
This ensures that the above sum must be less than or equal to 1, which ensures 
that $\mathbf{Z^*}(7MN+2N+2,p) \leq 1$ for that range of $p$.
\end{proof} 

Using Lemma 1, we interpret each element of $\mathbf{Z}(7MN+2N+2,p), p \in 
\{1,\ldots,(M+2)N\} $ as the marginal probability of offloading.  Furthermore, 
we note that 
the combination of (\ref{eq:SDP_k}) and (\ref{eq:SDPgeq0}) guarantees that all 
placement 
constraints in the recovered probabilities are also met---thus the assigned 
probability of a user offloading to a forbidden CAP is $0$. We can recover 
these probabilities for each user and denote them through the following vectors:
\begin{align*}
\mathbf{p}_i^x &= [p_{i0} \ldots p_{iM}], \quad\forall i \in \mathcal{N},\\
&= [\mathbf{Z}(7MN+2N+2,(M+1)(i-1)+1), \\
&\ldots, \mathbf{Z}(7MN+2N+2,(M+1)i)],  \\
\mathbf{P}_i^x &= [P_{i0} \ldots P_{iM}], \\
\text{s.t. } P_{ij} &= \frac{p_{ij}}{\sum_{k=0}^{M}p_{ik}}, \quad j \in 
\{0\}\cup\mathcal{M}, \\
P_i^\theta &= \frac{\mathbf{Z}(7MN+2N+2,(M+1)N+i)}{1-P_{i0}},
\end{align*}
where $\mathbf{p}_i^x$ is the original vector of recovered probabilities, 
$\mathbf{P}_i^x$ is the normalized $\mathbf{p}_i^x$ (done in case of any 
numerical imprecisions such that the total sum may not exactly equal 1), and 
$P_i^\theta$ is the conditional probability of cloud offloading given 
that there is transmission to a CAP.
The adjustment to $P_i^\theta$ is due to the fact that offloading to the cloud 
is impossible when local processing occurs.

The probabilistic mapping $\mathbf{P}_i^x$ is then used to produce the random 
vector $\mathbf{u}_i \in \{0,1\}^{M+1}$, which denotes a offloading 
decision:
\begin{align}
\mathbf{u}_i &= \mathbf{e}_j \text{ with probability } P_{ij}, \quad\forall 
j\in\{0\}\cup\mathcal{M,} \\
\Theta_i &= \begin{cases}
0, & \text{ with probability } 1 \text{ if } P_{i0} = 1, 1-P_i^\theta \text{ otherwise} \\
1, & \text{ with probability } 0 \text{ if } P_{i0} = 1, P_i^\theta \text{ otherwise}
\end{cases},
\end{align}
where $\mathbf{e}_j$ is the unit vector of size $M+1$ with dimension 
$j$.

Using this probability distribution, we generate $K$ i.i.d. trial 
solutions $\mathbf{x}_i^{(m)}$ and $\theta_i^{(m)}$ from the random vectors 
$\mathbf{u}_i$ and $\Theta_i$.  We solve problem (\ref{eq:resource_alloc}) in 
each one 
to obtain a set of 
offloading decisions the optimal set of offloading decisions from the trials.  
From this, 
we 
obtain a set of offloading decisions $\{\mathbf{x}_i^*\}$ and $\{\theta_i^*\}$ 
with a corresponding set of resource allocation decisions 
$\{\mathbf{c}_i^{u^*}\}, \{\mathbf{c}_i^{d^*}\}$, and $\{\mathbf{f}_i^{a^*}\}$, 
from which we select the one that gives the minimal objective
(\ref{eq:resource_alloc}).

The process stated above is outlined in Algorithm 1.  From observation, we have found that about $K = 10$ random trials are sufficient to produce near optimal system performance.

\begin{algorithm}[!t]
        \caption{MCAP Offloading Algorithm}
        \begin{algorithmic}[1]
                \State Solve SDP problem (\ref{eq:SDP}) to obtain matrix $\mathbf{Z^*}$.
                \State Recover $\mathbf{Z^*}$, and compute $\mathbf{P}_i^x$ 
                and $P_i^\theta$
                \For{$m=1$ to $M$}
                \State Generate $\mathbf{x}_i^{(m)}$ and $\theta_i^{(m)}$ from 
                $\mathbf{P}_i^x$ and $P_i^\theta$
                \State Solve (\ref{eq:resource_alloc}) and record the system cost 
                $J^{(m)}$ and resource 
                allocation decisions $\mathbf{c}_i^{u^{(m)}}, \mathbf{c}_i^{d^{(m)}}$, 
                and 
                $\mathbf{f}_i^{a^{(m)}}$
                \EndFor
                \State Find trial that produces the lowest cost: $m^*$ = $\arg\min J^{(m)}$
                \State Output: $\{\mathbf{x}_i^{*}\} =\{ \mathbf{x}_i^{(m^*)}\}$, 
                $\{\theta_i^{*}\} =\{\theta_i^{(m^*)}\}$ $\{\mathbf{c}_i^{u^{*}} \}= 
                \{\mathbf{c}_i^{u^{(m*)}}\}, \{\mathbf{c}_i^{d^{(m*)}}\} = 
                \{\mathbf{c}_i^{d^{*}}\}$, and $\{\mathbf{f}_i^{a^{*}}\} = 
                \{\mathbf{f}_i^{a^{(m*)}}\}$
        \end{algorithmic}
        \label{shareCAP}
\end{algorithm}

\section{Multi-user Mobile Cloud Offloading Game}
In this section, we model the interaction between mobile users and
the CAP as a mobile cloud offloading game, allowing selfish users 
agency 
over their offloading decisions.  We show that this game has an ordinal 
potential function, which implies that an NE exists and can be feasibly 
obtained.  We further demonstrate that our game is strategy-proof, thus giving 
users no incentive to provide false information to the system controller, 
allowing centralized computation of the NE.  We then propose the MCAP-NE 
algorithm, using the results 
from Section 4 and Algorithm 1 as an initial starting point in the computation 
of a NE, intuiting that a better starting point will yield an 
improvement in the number of iterations required to find the NE.  

\subsection{Game Formulation}
Consider a strategic form game

\begin{equation} \label{eq_Game_def_AF}
G_{\mathrm{MCO}}=(\mathcal{N},(\mathcal{A}_i)_{i\in \mathcal{N}},(u_i)_{i\in
\mathcal{N}}),
\end{equation}
where $\mathcal{N}=\{1,...,N\}$ is the player set containing all
mobile users, $\mathcal{A}_i = \{a_i\}$ is the strategy set for user $i$, and 
$u_i$ is the corresponding
cost function that user $i$ aims to minimize. Here, $u_i$ is a
function of the \textit{strategy profile}
$\mathbf{a}=(a_i,a_{-i})$, where $a_i\in\mathcal{A}_i,\
a_{-i}=(a_1,...,a_{i-1},...,a_{i+1},...,a_N)\in\mathcal{A}_{-i}=\prod_{j\neq
i}\mathcal{A}_j$.
Recall the strategy set $\{a_i\} = \{(\mathbf{x_i}, \theta_i)\}$ and individual 
cost functions (\ref{eq:ind_cost}). By choosing their offloading site $a_i$, 
user $i$ can decide where to
process their task to minimize their cost function $u_i$.

After receiving an offloading decision $a_i$ from all of the users, the CAP 
will assign communication and computation
resources to each user to minimize the overall system cost by
solving the convex resource allocation problem (\ref{eq:resource_alloc}).

\subsection{Structure Properties}
Given the selfish user assumption, we need to find an offloading 
decision that is stable, ensuring that users have no incentive to deviate from 
such a decision.  To 
this end, we consider the Nash Equilibrium \cite{osborne1994}:
\begin{Def}\label{def.NE} \it
The strategy profile $\mathbf{a}^*$ is a Nash equilibrium if
$u_i(a_i^*,a_{-i}^*)\leq u_i(a_i,a_{-i}^*)$, for any $a_i\in
\mathcal{A}_i,\ \forall i \in \mathcal{N}$.
\end{Def}
Definition \ref{def.NE} implies that, by employing strategies
corresponding to the NE, no player can decrease
their cost by unilaterally changing their own strategy. However, the NE
may not always exist, especially when the
game is not carefully formulated.  To ensure that our offloading game 
$G_{\text{MCO}}$ does have a NE, we demonstrate that it is an 
ordinal potential game \cite{monderer1996}:

\begin{Def}\label{def.OPG} \it
A strategic form game $G$ is an ordinal potential game (OPG) if
there exists an ordinal potential function $\phi: \prod_{
i}\mathcal{A}_i\to \mathbb{R}$ such that
\begin{align}\label{def.potential_function}
\mathrm{sgn}(u_i(a_i,a_{-i})&-
u_i(a_i',a_{-i}))\nonumber\\
&=\mathrm{sgn}(\phi(a_i,a_{-i})- \phi(a_i',a_{-i})),\quad\forall i,
\end{align}
where $a_i,a_i'\in \mathcal{A}_i, a_{-i}\in \mathcal{A}_{-i}.$
\end{Def}

\begin{Def}[Finite Improvement Property]\label{def.FIP} \it
A path in $G$ is a sequence $(\mathbf{a}[0],\mathbf{a}[1],...)$
where for every $k\geq 1$ there exists a unique player $i$ such that
$a_i[k]\neq a_i[k-1]\in \mathcal{A}_i$ while $a_{-i}[k]=
a_{-i}[k-1]$. It is an improvement
path if, for all $k\geq 1$,
$u_i(\mathbf{a}[k])<u_i(\mathbf{a}[k-1])$, where player $i$ is the
unique deviator at step $k$. $G$ has the finite improvement property
(FIP) if every improvement path in $G$ is finite.
\end{Def}

It is easy to see the following relationship between an OPG and the FIP and NE 
[12]:

\begin{Lem}\label{thm.OPG} \it
Every OPG with finite strategy sets possesses at least one
pure-strategy NE and has the FIP.
\end{Lem}

We now show that $G_{\mathrm{MCO}}$ is indeed an OPG.

\begin{Prop}\label{prop.G_MCO} \it
The proposed mobile cloud offloading game $G_{\mathrm{MCO}}$ is an
OPG and,
therefore, it always has an NE and the FIP.
\end{Prop}
\begin{proof}
We first construct the potential function
\begin{align}\label{potential_function}
 \phi(\mathbf{a})=&\bigg[\sum_{i=1}^N\alpha_i(E_{l_{i}}x_{l_{i}}+E_{A_i}x_{a_{i}}+E_{C_i}x_{c_{i}})\nonumber\\
&+\max_{i}\{T_{L_{i}}+T_{A_{i}}+T_{C_{i}}\}\bigg],
\end{align}
and note that our potential function is equal to the system objective.

Define $\mathbf{E}_i\triangleq
\alpha_i(E_{l_{i}}+E_{A_i}+E_{C_i})$ and
$\mathbf{T}_i\triangleq T_{L_{i}}+T_{A_{i}}+T_{C_{i}}$. Given two
different strategy profiles $\mathbf{a}=(a_i,a_{-i})$ and
$\mathbf{a}'=(a_i',a_{-i})$, where only user $i$ chooses different
strategies $a_i$ and $a_i'$, respectively, the difference between
$\phi(\mathbf{a})$ and $\phi(\mathbf{a}')$ is
\begin{align}
 \phi(\mathbf{a})-\phi(\mathbf{a}')&=\ \mathbf{E}_i+\sum_{j\neq i}\mathbf{E}_j+\max\{\mathbf{T}_i,\max_{j\neq i}\{\mathbf{T}_j\}\}\nonumber\\
 &\ \ \ \  -\mathbf{E}_i'-\sum_{j\neq 
 i}\mathbf{E}_j-\max\{\mathbf{T}_i',\max_{j\neq i}\{\mathbf{T}_j\}\},\nonumber\\
 &=\ \mathbf{E}_i+\max\{\mathbf{T}_i,\max_{j\neq i}\{\mathbf{T}_j\}\}\nonumber\\
 &\ \ \ \ -\mathbf{E}_i'-\max\{\mathbf{T}_i',\max_{j\neq 
 i}\{\mathbf{T}_j\}\},\nonumber\\
 &=\ u_i(\mathbf{a})-u_i(\mathbf{a}').\nonumber
\end{align}
Since $\phi(\mathbf{a})$ in \eqref{potential_function} satisfies the
condition of the potential function of an OPG defined in
\eqref{def.potential_function}, it is a potential function
of $G_{\mathrm{MCO}}$. Therefore, $G_{\mathrm{MCO}}$ is an OPG by Definition 2
\end{proof}

\subsection{Strategy-Proofness}
Note that the computation of an NE in $G_{\text{MCO}}$ requires accurate task 
information from all
users. We thus must show that there is no incentive for any
user to provide false task information (i.e. data sizes, required number of 
CPU cycles, and relative weight $\alpha_i$). 

First, we note that if a user is found by the CAP to provide false information, 
it will be
prohibited from participating in the system, so no user will both
provide false information and offload its task to a CAP in the
same round, when its deceit will be noticed by the CAP. Thus, the deceitful 
user's strategy must be $x_{i0} = 1$ in that round.  Hence, at any NE, 
the delay of 
the round will be $\max\{T_i^{l'}, T_{-i}'\}$, where $T_i^{l'}$ is the false 
local processing time for user $i$, and $T_i^{l'}$ is the maximum delay of all 
the other users tasks given the current strategy $\mathbf{a}'^{*}$.  The system 
cost or potential function at this point is
\begin{equation}
        \phi(\mathbf{a}'^{*}) = \alpha_i'E_i^{l'} + \sum_{j\neq 
        i}\alpha_jE_j' + 
        \max\{T_i^{l'}, T_{-i}'\}. \label{eq:strategy}
\end{equation}
where $E_i^{l'}$ is the local processing energy for user $i$, and $E_j'$ is 
the energy consumption for user $j$ at strategy $\mathbf{a}'^{*}$.

If user $i$ does not participate in the game however, the system cost for the 
remaining users at NE $a_{-i}^*$ is

\begin{equation}
        \phi_{-i}(a_{-i}^*) = \sum_{j\neq i}\alpha_jE_j + T_{-i},
\end{equation}
with $T_{-i}$ being the delay for that round.  Note that $\sum_{j\neq 
i}\alpha_jE_j' + T_{-i}' = \phi_{-i}(a_{-i}'^{*}) \geq 
\phi_{-i}(a_{-i}^*) = 
\sum_{j\neq i}\alpha_jE_j + T_{-i}$ where the equality holds when $a_{-i}'^{*} 
= a_{-i}^*$ (and thus $E_j' = E_j$ and $T_{-i}' = T_{-i}$.

We now show that $T_{-i} \leq \max\{T_i^{l'}, T_{-i}'\}$ always.  If $T_{-i} 
> T_{-i}' \geq T_i^{l'}$, then:
\begin{align*}
\phi(\mathbf{a}'^{*}) &= \alpha_i'E_i^{l'} + \sum_{j\neq 
        i}\alpha_jE_j' + 
\max\{T_i^{l'}, T_{-i}'\}, \\
&= \alpha_i'E_i^{l'} + \sum_{j\neq i}\alpha_jE_j' + T_{-i}', \\
&> \alpha_i'E_i^{l'} + \sum_{j\neq i}\alpha_jE_j + T_{-i}, \\
&= \alpha_i'E_i^{l'} + \sum_{j\neq i}\alpha_jE_j + \max\{T_i^{l'}, 
T_{-i}\},
\end{align*}
which is contradictory to (\ref{eq:strategy}).  If $T_{-i} > T_i^{l'} \geq 
T_{-i}'$ then
\begin{align*}
        \phi(\mathbf{a}'^{*}) &= \alpha_i'E_i^{l'} + \sum_{j\neq 
        i}\alpha_jE_j' + 
        \max\{T_i^{l'}, T_{-i}'\}, \\
        &= \alpha_i'E_i^{l'} + \sum_{j\neq i}\alpha_jE_j' + T_{i}^{l'}, \\
        &> \alpha_i'E_i^{l'} + \sum_{j\neq i}\alpha_jE_j + T_{-i}', \\
        &> \alpha_i'E_i^{l'} + \sum_{j\neq i}\alpha_jE_j + T_{-i},\\
        &= \alpha_i'E_i^{l'} + \sum_{j\neq i}\alpha_jE_j + \max\{T_i^{l'}, 
        T_{-i}\},
\end{align*}
which is also contradictory to (\ref{eq:strategy}).  Thus, by providing false 
information and 
processing its tasks locally, user $i$ lengthens the round-time and incurs a 
higher cost compared with processing its task locally without participating in 
the game.

Therefore, no user participating in the game has any incentive not 
to be untruthful to the system controller.

\subsection{MCAP-NE Offloading Algorithm}

\begin{algorithm}[!t]
  \caption{MCAP-NE Offloading Algorithm}
  \begin{algorithmic}[1]
    \State Obtain an initial strategy profile $\mathbf{a}[0]$ and corresponding 
    resource allocation decisions $\{\mathbf{c}_i^u\}, \{\mathbf{c}_i^d\}$, and 
    $\{\mathbf{f}_i^a\}$ from Algorithm 1.
    \State Set $\mathrm{NE}=\mathrm{False}$ and $k=0$.
    \While{$\mathrm{NE}==\mathrm{False}$}
    $\mathrm{flag}=0$ and $i=1$;
    \While{$flag==0$ and $i\leq N$}
    \State Calculate $\{\mathbf{c}_i^{u^{*}} \}, \{\mathbf{c}_i^{d^{*}}\}$, and 
    $\{\mathbf{f}_i^{a^{*}}\}$ for $(a_i',a_{-i}[k])$, for all $a_i'\in
    \mathcal{A}_i$;
    \If{$u_i(\mathbf{a}[k])> u_i(a_i',a_{-i}[k]),a_i'\in \mathcal{A}_i$}
    \State Set $a_i[k+1]=a_i'$,
    $a_{-i}[k+1]=a_{-i}[k]$;
    \State Set $\mathbf{a}[k+1]=(a_i[k+1],a_{-i}[k+1])$, $\mathrm{flag}=1$;
    \State $k=k+1$;
    \ElsIf{$i==N$}
    \State Set $\mathrm{flag}=1$, $\mathrm{NE}=\mathrm{True}$;
    \Else
    \State $i=i+1$;
    \EndIf
    \EndWhile
    \EndWhile
    \State Output: the NE $\mathbf{a}^*$ of $G_{\mathrm{MCO}}$ and the 
    corresponding resource allocation $\{\mathbf{c}_i^{u^*}\}, 
    \{\mathbf{c}_i^{d^*}\}$, and $\{\mathbf{f}_i^{a^*}\}$.
  \end{algorithmic}
  \label{algorithm}

\end{algorithm}

In this section we propose a mobile cloud offloading algorithm
based on FIP to find an NE of $G_{\text{MCO}}$. Since the CAP
has all of the necessary information from the mobile users
and needs to compute the communication and computational
resources to the offloading users, we propose a centralized
approach to computing the NE.

In our solution method, the CAP first initiates a starting strategy profile 
$\mathbf{a}[0]$
containing a set of hypothetical offloading decisions for all users. Based on
$\mathbf{a}[0]$, it obtains the optimal resources allocation by solving 
(\ref{eq:resource_alloc}).
Then, the CAP takes an arbitrarily ordered list of the users and
examines each user's strategy set one-by-one. Once it
finds a user $i$ who can reduce their individual cost by switching from 
strategy $a_i[0]$ to another strategy $a_i'\in \mathcal{A}_i$
(with $a_{-i}[0]$ remaining constant), it updates the strategy profile
from $\mathbf{a}[0]$ to $\mathbf{a}[1]$ where $a_i[1]=a_i'$ and
$a_{-i}[1]=a_{-i}[0]$ (subject to placement constraints). The CAP repeats the 
same procedure to find an
improvement path $(\mathbf{a}[0],\mathbf{a}[1],\mathbf{a}[2],...)$
of $G_{\mathrm{MCO}}$.  Because $G_{\text{MCO}}$ is an ordinal potential game, 
this improvement path is guaranteed to terminate at an NE.

While any initial starting point will eventually lead to a NE through an 
improvement path, the choice of initial point may have an effect on the length 
of the improvement path.  Since each step of the improvement path method may 
require a nontrivial amount of time (due to the combinatorial nature of the 
finite improvement method), reducing the number of iterations of this method 
can greatly 
improve the computational time required to compute the NE.  Thus, we propose 
using the result of the 
MCAP 
method detailed in Section 4 as our initial point.  Because the result from 
the MCAP algorithm is substantially closer to optimal than a random 
starting 
point, we expect that using the MCAP solution as our initial point 
will reduce 
the 
number of iterations required to compute the NE, which we have confirmed in 
simulation.

The details of the proposed algorithm, which we term MCAP-NE, are given in
Algorithm \ref{algorithm}.



\section{Simulation Results}

In this section, we detail the simulation results of MCAP and MCAP-NE. 
We first consider a system without any user placement constraints, 
subjecting 
the system through a series of parameter changes. We observe the system cost 
and the number of 
iterations required using the different optimization methods.  We then consider 
the addition of placement constraints in the system.

\subsection{Default Parameters}
\begin{table}[!t]
        \caption{Default Simulation Parameter Values} \centering
        
        \begin{tabular}{|r| l|}
                \hline 
                \textbf{Parameter}& \textbf{Default Value} \\
                \hline 
                Processing cycles per byte & 1900 \\
                Minimum input data size & 10 MB \\
                Maximum input data size & 30 MB \\
                Minimum output data size & 1 MB \\
                Maximum output data size & 3 MB \\
                Number of CAPS & 2 \\
                $\alpha_i$ & 0.5 s/J, $\quad\forall i$ \\
                $\beta_i$ & 1.7 $\times 10^{-7}$ \\
                $C_{\text{UL}}^j, C_{\text{DL}}^j$ & 20 MHz, $\quad\forall j$ \\
                $C_{\text{Total}}^j$ & 40 MHz, $\quad\forall j$ \\
                Minimum $\eta_{ij}^u, \eta_{ij}^d$ & 2 b/s/Hz, $\quad\forall i,j$ \\
                Maximum $\eta_{ij}^u, \eta_{ij}^d$ & 5 b/s/Hz, $\quad\forall i,j$\\
                Local CPU speed & $2.39 \times 10^9$ cycles/s \\
                CAP CPU speed & $5\times10^9$ s/bit \\
                Cloud CPU speed & $7.5\times10^9$ s/bit \\
                Tx and Rx energy & $1.42\times 10^{-7}$ J/bit \\
                CAP to cloud transmission rate & $R_{ac}=15$ Mpbs \\
                \hline
        \end{tabular}
        \label{table_parameters}
        \vspace{-0.5cm}
\end{table}
 We utilize the x264 CBR encoding application, 
which requires 1900 cycles/byte 
\cite{miettinen2010}.  The input and output data
sizes of each task are assumed to be uniformly distributed from
$10$MB to $30$MB and from $1$MB to $3$MB, respectively. We set by default the 
number of 
users $N = 10$,
number of CAPs $M = 2$, $\alpha = 0.5$ s/J, and $\beta = 1.7 \times 10^{-7}$ 
J/bit. The
bandwidths at each CAP are $C_{\text{UL}}^j = C_{\text{DL}}^j = 20$ MHz, and 
$C_{\text{Total}}^j = 40 MHz$ 
for 
each CAP, and the spectral efficiencies $\eta_{ij}^u = \eta_{ij}^d$ are 
uniformly distributed between 2 and 5 b/s/Hz, which correspond to typical WiFi 
communication settings.  We use an iPhone X mobile device 
with a CPU speed of $2.39 \times 10^9$ cycles/s, leading to a local computation 
time of 
$9.93\times10^{-8}$ s/bit \cite{iPhone}, and 
adopt a CPU rate of $5\times10^9$ cycles/s at the CAPs, and $7.5\times10^9$ 
cycles/s at the cloud. The transmission and receiving energy per bit at each
mobile device are both $1.42\times 10^{-7}$ J/bit as indicated in
Table 2 in \cite{miettinen2010}.  For offloading a task to the
cloud, the transmission rate is $R_{ac}=15$ Mpbs. Also, we set the
cloud utility cost $C_{c_{i}}$ to be the same as that of the input
data size $D_{\textrm{\text{in}}}(i)$.  These parameter values are listed in 
Table II.

\subsection{Impact of Different Parameters}
  We run the simulation through 50 rounds, with the input 
and output data size of each task being
independently and identically generated, and plot the averaged total
system cost.  We compare the performance of MCAP and MCAP-NE 
against the optimal solution (obtained through exhaustive search), a random 
mapping of offloading decisions, and the NE 
results with random mapping as the starting point.  

We consider the above system settings while varying a single parameter, in 
order to demonstrate our solution's superior performance under a variety of 
settings.
We note that in each of these figures, both MCAP and MCAP-NE incur costs close 
to that of the optimal solution, despite the great strategy 
space available, while MCAP-NE consistently improves slightly upon MCAP.  This 
shows that our solution is highly reliable 
in 
recovering 
near-optimal solutions to the multi-CAP optimization problem.  By contrast, 
\textit{random mapping} consistently performs far worse.  Figure 
\ref{fig:multiCAP} shows that the cost 
decreases with the number of available CAPs, which is expected given the 
additional resources that each CAP provides to the users.
Figure \ref{fig:multiuser}, shows that the system cost increases 
with the number of users, which is as expected given the added competition of 
resources these users produce.  Figures 
\ref{fig:multialpha} and \ref{fig:multibeta} both show that the system cost 
increases with an increase in the cost weights $\alpha$ and $\beta$, though the 
increase levels off with $\beta$.  This is because increasing $\beta$ would 
have no effect when it is already too high for the users to utilize cloud 
processing.

\begin{figure}
        \centering
        \begin{subfigure}[b]{0.5\textwidth}
                \centering
                \includegraphics[width=\textwidth]{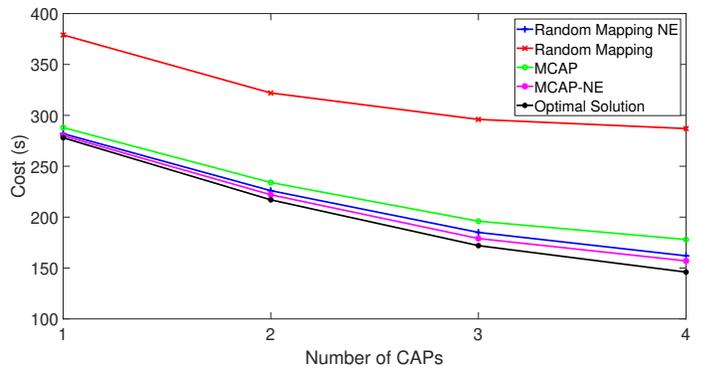}
                \caption{Cost against number of CAPs.}
                \label{fig:multiCAP}
        \end{subfigure}
        \begin{subfigure}[b]{0.5\textwidth}
                \centering
                \includegraphics[width=\textwidth]{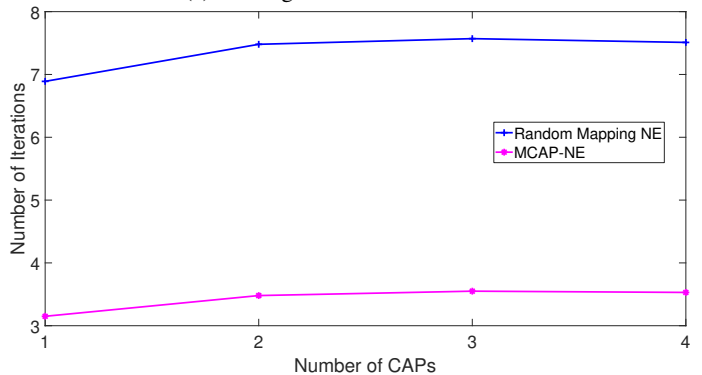}
                \caption{Number of Iterations required against CAPs.}
                \label{fig:numITCAPS}
        \end{subfigure}
        \caption{Impact of the Number of CAPs.}
\end{figure}

\begin{figure}
        \centering
        \begin{subfigure}[b]{0.5\textwidth}
                \includegraphics[width=\textwidth]{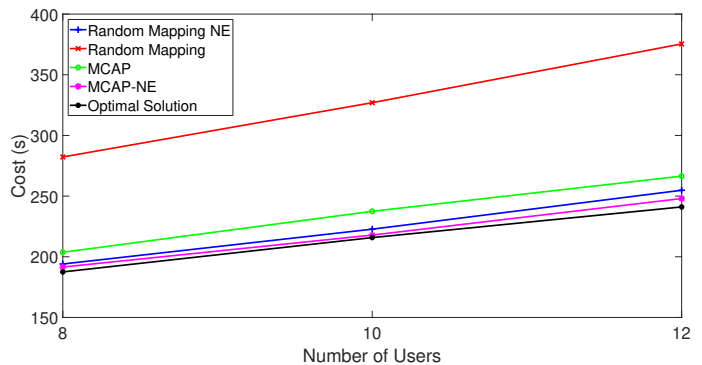}
                \caption{Impact of the number of Users.}
                \label{fig:multiuser}
        \end{subfigure}
        \hfill
        \begin{subfigure}[b]{0.5\textwidth}
                \includegraphics[width=\textwidth]{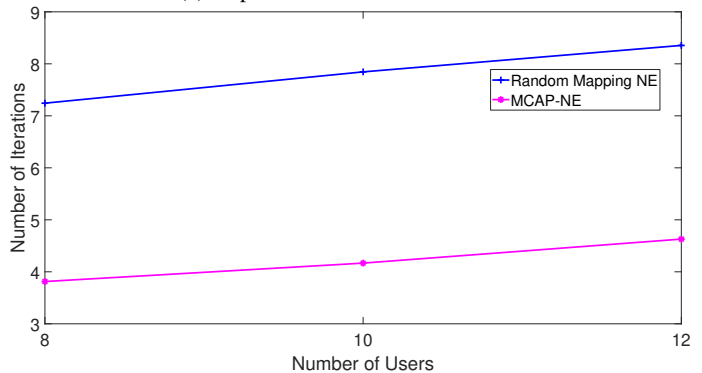}
                \caption{Number of Iterations against the number of Users.}
                \label{fig:numITUSERS}
        \end{subfigure}
        \caption{Impact of Number of Users.}
\end{figure}

\begin{figure}
        \centering
        \begin{subfigure}[b]{0.5\textwidth}
                \centering
                \includegraphics[width=\textwidth]{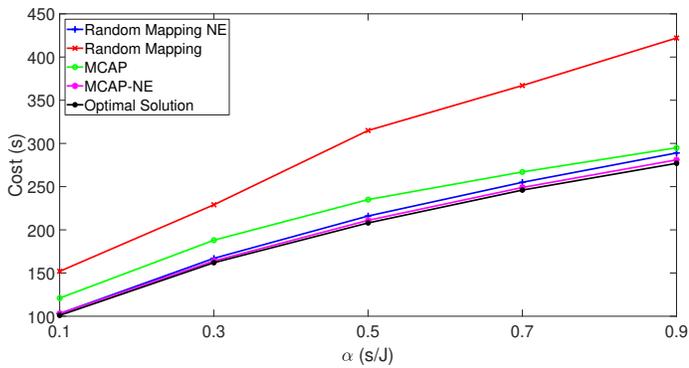}
                \caption{Cost against $\alpha$}
                \label{fig:multialpha}
        \end{subfigure}
        \begin{subfigure}[b]{0.5\textwidth}
                \centering
                \includegraphics[width=\textwidth]{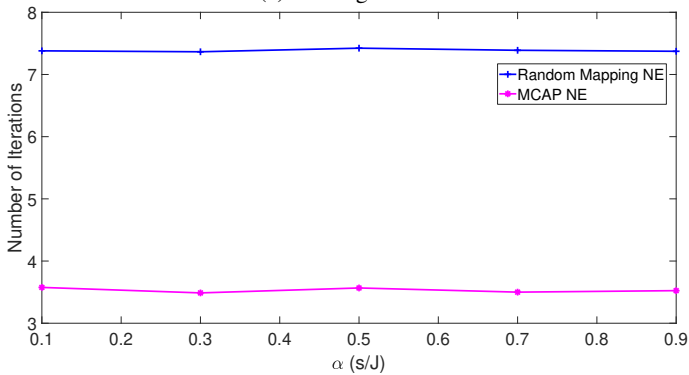}
                \caption{Number of Iterations against $\alpha$.}
                \label{fig:numITALPHA}
        \end{subfigure}
        \caption{Impact of $\alpha$.}
\end{figure}

\begin{figure}
        \centering
        \begin{subfigure}[b]{0.5\textwidth}
                \centering
                \includegraphics[width=\textwidth]{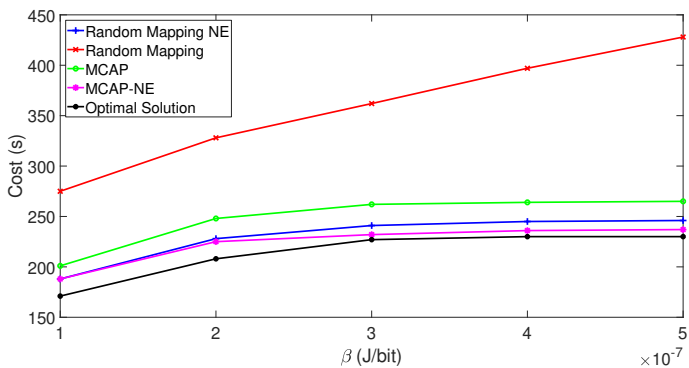}
                \caption{Cost against $\beta$}
                \label{fig:multibeta}
        \end{subfigure}
        \begin{subfigure}[b]{0.5\textwidth}
                \centering
                \includegraphics[width=\textwidth]{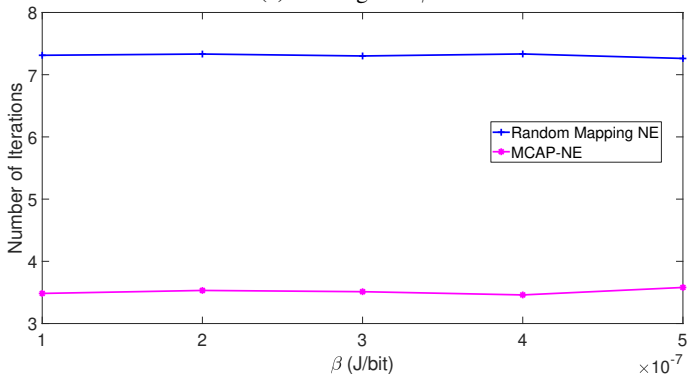}
                \caption{Number of Iterations required $\beta$.}
                \label{fig:numITBETA}
        \end{subfigure}
        \caption{Impact of $\beta$.}
\end{figure}

Figures \ref{fig:numITCAPS}, \ref{fig:numITUSERS},  \ref{fig:numITALPHA}, and 
\ref{fig:numITBETA} show the number of 
iterations required to compute the NE against the number of CAPs, the number of 
users, $\alpha$, and
$\beta$,
respectively, with either a random starting point or an MCAP starting point.  
In all of these 
figures, we see that number of iterations required to 
obtain the NE is 
more than doubled when using a random starting point as opposed to the 
MCAP-NE, confirming the performance benefit of using 
MCAP for an initial starting point.  While the addition of CAPs 
does 
not noticeably increase the number of iterations required, the number of users 
does have a discernible effect due to the additional number of possible users 
who may find an improvement.  As expected, $\alpha$ 
and $\beta$ have no discernible effect as those parameters do not affect the 
size of the strategy space.

\subsection{Placement Constraints}

We now study a system in the presence of placement constraints.  Here, we 
consider a system of 12 users by default, randomly assigning as a placement 
constraint either one of the CAPs or the empty set with equal probability.  
All of the other parameters of the system are the same as in Section 6.1.

\begin{figure}[!t]
        \centering
        \includegraphics[width=\linewidth]{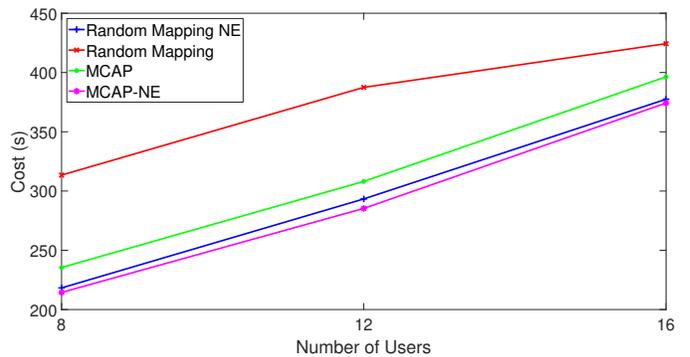}
        \caption{Cost against the number of users with placement constraints}
        \label{fig:numUsersPlacement}
\end{figure}

\begin{figure}[!t]
        \centering
        \includegraphics[width=\linewidth]{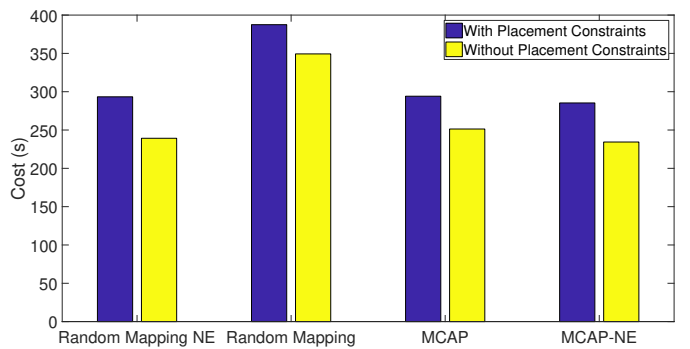}
        \caption{Cost with placement constraints and without placement constraints.}
        \label{fig:placement}
\end{figure}

Figures \ref{fig:numUsersPlacement} shows the system cost against the number of 
users under the system with placement constraints.  We note that the 
behaviour here is similar to the behaviour in the system 
without placement constraints --- the cost is increasing with respect to the 
number of users, and relative performance between offloading methods remain 
close to the original case.  Figure \ref{fig:placement} compares a system with 
placement constraints with a system without placement constraints.  We observe 
that the system with placement constraints requires a higher cost, as 
beneficial offloading decisions that would otherwise be desirable cannot be 
made.

\section{Conclusion}
A multi-user mobile cloud computing system with multiple CAPs has been
considered, in which each mobile user has a task to be processed either 
locally, at one of the CAPs, subject to the user's placement constraint, or at 
a remote cloud.  We solve the non-convex optimization problem through
two approaches: in MCAP we formulate the problem as a
QCQP and use an SDP relaxation to arrive at a heuristic solution; while in 
MCAP-NE, we accommodate selfish users by leveraging the finite improvement 
property of anordinal potential game to find an NE.  Simulation results
show near optimal performance of our methods as well as the
utility of the SDP starting point in substantially reducing the number of 
iterations required to compute the NE.

\bibliographystyle{IEEEtran}
\bibliography{mcloud}

\end{document}